\documentclass{article} 

\usepackage{amsthm,amsfonts,amsmath,amscd,amssymb,graphicx,tikz-cd,mathtools,mathrsfs,authblk,url}  

\newtheorem{theorem}{Theorem}
\newtheorem{lemma}[theorem]{Lemma}
\newtheorem{proposition}[theorem]{Proposition}

\def\R{{\mathbb R}}

\begin{document}

\author{Colin McSwiggen}
\date{}

\title{A new proof of Harish-Chandra's integral formula}

\maketitle

\begin{abstract}
We present a new proof of Harish-Chandra's formula $$\Pi(h_1) \Pi(h_2) \int_G e^{\langle \mathrm{Ad}_g h_1, h_2 \rangle} dg = \frac{ [ \! [ \Pi, \Pi ] \!] }{|W|} \sum_{w \in W} \epsilon(w) e^{\langle w(h_1),h_2 \rangle},$$ where $G$ is a compact, connected, semisimple Lie group, $dg$ is normalized Haar measure, $h_1$ and $h_2$ lie in a Cartan subalgebra of the complexified Lie algebra, $\Pi$ is the discriminant, $\langle \cdot, \cdot \rangle$ is the Killing form, $[ \! [ \cdot, \cdot ] \!]$ is an inner product that extends the Killing form to polynomials, $W$ is a Weyl group, and $\epsilon(w)$ is the sign of $w \in W$.

The proof in this paper follows from a relationship between heat flow on a semisimple Lie algebra and heat flow on a Cartan subalgebra, extending methods developed by Itzykson and Zuber for the case of an integral over the unitary group $U(N)$.  The heat-flow proof allows a systematic approach to studying the asymptotics of orbital integrals over a wide class of groups.
\end{abstract}

\section{Introduction}

Harish-Chandra's integral formula states \begin{equation} \label{eqn:hc} \Pi(h_1) \Pi(h_2) \int_G e^{\langle \mathrm{Ad}_g h_1, h_2 \rangle} dg = \frac{ [ \! [ \Pi, \Pi ] \!] }{|W|} \sum_{w \in W} \epsilon(w) e^{\langle w(h_1),h_2 \rangle},\end{equation} where $G$ is a compact, connected, semisimple Lie group, $dg$ is normalized Haar measure, $h_1$ and $h_2$ lie in a Cartan subalgebra of the complexified Lie algebra, $\langle \cdot, \cdot \rangle$ is the Killing form, $W$ is a Weyl group, and $\epsilon(w)$ is the sign of $w \in W$.  The polynomial $\Pi$ and bracket notation $ [ \! [ \Pi, \Pi ] \!]$ are defined below in (\ref{eqn:pi-def}) and (\ref{eqn:bracket-def}) respectively.  We explain all notation in detail in Section \ref{sec:preliminaries}.

This identity was first proved by Harish-Chandra in his 1957 paper on invariant differential operators \cite{HC}.  The importance of integrals such as (\ref{eqn:hc}) for mathematical physics was first noted in 1980 by Itzykson and Zuber \cite{IZ}, who independently discovered the formula for the case of an integral over the unitary group.  The unitary integral has become known as the Harish-Chandra-Itzykson-Zuber (HCIZ) integral and has applications in random matrix theory, representation theory, combinatorics, and statistics.  It has also been shown that this integral can be understood in terms of enumerative geometry \cite{GGN}.  The HCIZ integral is typically written as
\begin{equation} \label{eqn:hciz} \int_{U(N)} e^{\mathrm{tr} (AUBU^\dagger)} dU = \left( \prod_{p=1}^{N-1}p! \right) \frac{\det(e^{a_i b_j})_{i,j = 1}^N}{\Delta(A) \Delta(B)} \end{equation}
where $U(N)$ is the group of $N$-by-$N$ unitary matrices, $A$ and $B$ are fixed $N$-by-$N$ diagonal matrices with eigenvalues $a_1 < \hdots < a_N$ and $b_1 < \hdots < b_N$ respectively, and $$\Delta(A) = \prod_{i < j} (a_j - a_i)$$ is the Vandermonde determinant.  Although $U(N)$ is not semisimple, (\ref{eqn:hciz}) follows from (\ref{eqn:hc}) by decomposing $U(N) = SU(N) \rtimes U(1)$ and integrating separately on the factors.  A general procedure for using (\ref{eqn:hc}) to evaluate integrals over compact groups that may be neither semisimple nor connected is explained in \cite{CM}.

Harish-Chandra's proof of (\ref{eqn:hc}) relies on algebraic constructions that he develops for the purpose of studying Fourier transforms on a semisimple Lie algebra, and the theorem appears as an auxiliary result within a large project in harmonic analysis and representation theory.  These may be reasons that analogous integrals over groups other than $U(N)$ have received relatively little attention from mathematical physicists,\footnotemark \footnotetext{The papers \cite{GZ, AG, BBMP}, among others, also treat integrals of the form (\ref{eqn:hciz}) over the orthogonal group $O(N)$ and compact symplectic group $USp(2N)$.  However, those integrals differ from (\ref{eqn:hc}), as they take the matrices $A$ and $B$ to be symmetric in the orthogonal case or self-dual quaternionic in the symplectic case. In contrast, to get an integral of the form (\ref{eqn:hc}) we would have to take $A$ and $B$ to be {\it anti-}symmetric or anti-self dual.} although Forrester, Ipsen, Liu and Zhang have recently used the Harish-Chandra integrals over $O(N)$ and $USp(2N)$ to derive eigenvalue probability densities for certain matrix product ensembles \cite{FILZ}.

In this paper, we prove Harish-Chandra's formula by relating the heat flow on a semisimple Lie algebra to the heat flow on a Cartan subalgebra.  The proof generalizes the methods employed by Itzykson and Zuber in their first proof of the HCIZ formula in \cite{IZ}.  The main purpose of our proof is to provide insight into the asymptotics of (\ref{eqn:hc}) as the rank of $G$ increases. Leading-order asymptotics for (\ref{eqn:hciz}) as $N \to \infty$ were formally computed by Matytsin \cite{AM} and rigorously justified by Guionnet and Zeitouni \cite{GZ, AG}, while Bun et al.~have reformulated these results in terms of particle trajectories concentrating around an instanton \cite{BBMP}.

The central result in \cite{AM, GZ, AG} on the large-$N$ asymptotics of the HCIZ integral is an expression for the leading-order contribution in terms of a particular solution to the complex Burgers' equation.  It has been shown that this phenomenon can be understood in terms of a relationship between the HCIZ integral and the Calogero--Moser system \cite{GM}.  In fact, Harish-Chandra integrals arise naturally in the study of the quantum Calogero--Moser system, where the homomorphism of invariant differential operators that Harish-Chandra uses to derive (\ref{eqn:hc}) has been shown to provide a quantum Hamiltonian reduction of the algebra of differential operators on a reductive Lie algebra \cite{LS, PE}.  We elaborate on the relationship between Harish-Chandra integrals and Calogero--Moser systems in Section \ref{subsec:cm_relationship} below.

While it suffices for the purposes of this paper to relate heat flow on a Lie algebra to heat flow on a Cartan subalgebra, one could obtain richer information by studying the relationship between Brownian motion on the full algebra and Brownian motion confined to a Weyl chamber.  This latter process is an analogue of Dyson Brownian motion and has been studied by Grabiner in \cite{DG}.  The arguments regarding asymptotics for the $U(N)$ integral in \cite{GZ, AG} rely on a large-deviations principle for Dyson Brownian motion, suggesting that a similar investigation of processes on general Weyl chambers could prove useful in the study of integrals over other groups.

There are other known proofs of (\ref{eqn:hciz}), besides the heat-equation approach. These include a symplectic geometry proof based on a stationary-phase approximation that is shown to be exact using the Duistermaat--Heckman theorem \cite{DH} and a representation-theoretic proof that proceeds by writing the integral in terms of irreducible characters \cite{IZ}.  In fact in \cite{DH} Duistermaat and Heckman already observed that their localization result implies the general formula (1).  We give a detailed presentation of this derivation in Section \ref{subsec:symplectic_proof}, and in Section \ref{subsec:repthy_proof} we show how the integral formula follows from its interpretation in character theory.

Section \ref{sec:preliminaries} below introduces the necessary notation and background to formally state the integral formula (\ref{eqn:hc}) and provides a brief overview of the proof strategy.  Section \ref{sec:hc_heat_eqn_proof} consists of the heat-flow proof of the integral formula. Section \ref{sec:concl_remarks} discusses the relationship between Harish-Chandra integrals and Calogero--Moser systems and presents the two further proofs of the integral formula from the perspectives of symplectic geometry and representation theory.  The expository paper \cite{CM} contains a more detailed discussion of the formula and its contemporary significance, along with a modern presentation of Harish-Chandra's original proof and explicit formulae for the integrals over the other compact classical groups.

\section{Preliminaries and statement of the theorem} \label{sec:preliminaries}

Let $G$ be a compact, connected, semisimple real Lie group of rank $N$ with normalized Haar measure $dg$, $\mathfrak{g}_0$ its Lie algebra, and $\mathfrak{g} = \mathfrak{g}_0 \otimes \mathbb{C}$ the complexification of $\mathfrak{g}_0$.  Let $\langle \cdot, \cdot \rangle$ be the Killing form on $\mathfrak{g}_0$, which extends by linearity to $\mathfrak{g}$.

Let $\mathfrak{h}_0 \subset \mathfrak{g}_0$ be a Cartan subalgebra, $\mathfrak{h} = \mathfrak{h}_0 \otimes \mathbb{C} \subset \mathfrak{g}$ its complexification, and $W$ the Weyl group acting on $\mathfrak{h}$. For $w \in W$, let $\epsilon(w)$ be the {\it sign} of $w$, equal to 1 if $w$ is generated by an even number of reflections and $-1$ if $w$ is generated by an odd number of reflections.  If we consider $w$ as a linear operator acting on $\mathfrak{h}$, then $\epsilon(w) = \det w$.

Fix a choice $\alpha_1, \hdots, \alpha_r$ of the positive roots of $\mathfrak{h}$, and define the {\it discriminant} $\Pi: \mathfrak{h} \to \mathbb{C}$ to be the polynomial function \begin{equation} \label{eqn:pi-def} \Pi(h) = \prod_{i=1}^r \alpha_i(h).\end{equation}  Geometrically, $\Pi(h)$ is the square root of the volume of the adjoint orbit of $h$, computed using the metric induced by the Killing form metric on $\mathfrak{g}_0$.  It can be shown (see e.g. \cite[ch.~3, Corollary 3.8]{SH}) that $\Pi$ is {\it skew} with respect to the action of the Weyl group, in the sense that $\Pi(w(h)) = \epsilon(w) \Pi(h)$.

Let $(z_j)_{j=1}^{\dim \mathfrak{g}}$ be complex coordinates on $\mathfrak{g}$ such that $z_j = x_j + i y_j$ where $(x_j), (y_j)$ are coordinates on $\mathfrak{g}_0$.  We can write a polynomial function $p$ on $\mathfrak{g}$ as $p(z) = \sum c_\beta z^\beta$, where the sum runs over multi-indices $\beta$. We can then associate to $p$ the differential operator \begin{equation} \label{eqn:assoc-diffop} p(\partial) = \sum_\beta c_\beta \frac{\partial^{|\beta|}}{ \partial z^\beta} = \sum_\beta c_\beta \prod_{j=1}^{\dim \mathfrak{g}} \left( \frac{\partial}{\partial x_j} - i \frac{\partial}{\partial y_j} \right)^{\beta_j}.\end{equation}  The coefficients $c_\beta$ depend on the choice of coordinates $(z_i)$, but the operator $p(\partial)$ does not.  We may always identify a polynomial function on $\mathfrak{g}$ with its restriction to $\mathfrak{g}_0$, and we can consider $p(\partial)$ as a differential operator on $\mathfrak{g}_0$ by defining, for $f \in C^\infty(\mathfrak{g}_0)$, \begin{equation} \label{eqn:diffop-restr} p(\partial) f = \sum_\beta c_\beta \frac{\partial^{|\beta|} f}{\partial x^\beta}, \end{equation} that is by differentiating only with respect to the real coordinates on $\mathfrak{g}_0$.

Finally, if $p, q$ are polynomial functions on $\mathfrak{g}$, define\footnotemark \footnotetext{See \cite[ch.~3 \textsection 1]{SH} for a proof that $[ \! [ \cdot, \cdot ] \! ]$ in fact defines an inner product on the symmetric algebra of formal (commutative) polynomials over $\mathfrak{g}$, which restricts to the Killing form on $\mathfrak{g}$.}  \begin{equation} \label{eqn:bracket-def} [ \! [ p, q ] \! ] = p(\partial) q(z) \big |_{z = 0}. \end{equation}

We are now ready to state Harish-Chandra's integral formula.
\begin{theorem} \label{thm:hc} For all $h_1, h_2 \in \mathfrak{h}$, $$ \Pi(h_1) \Pi(h_2) \int_G e^{\langle \mathrm{Ad}_g h_1, h_2 \rangle} dg = \frac{ [ \! [ \Pi, \Pi ] \!] }{|W|} \sum_{w \in W} \epsilon(w) e^{\langle w(h_1),h_2 \rangle}. $$
\end{theorem}

In his original proof of Theorem \ref{thm:hc}, Harish-Chandra starts by showing a general form of Proposition \ref{prop:radial-laplacian} below on the radial part of an invariant differential operator, which implies that the function $$h_1 \mapsto \Pi(h_1) \int_G e^{\langle \mathrm{Ad}_g h_1, h_2 \rangle} dg$$ is a joint eigenfunction of all differential operators on $\mathfrak{h}_0$ of the form $q(\partial)$, where $q$ is a $W$-invariant polynomial.  He then proves that the algebra of all polynomial functions on $\mathfrak{h}_0$ is generated over the $W$-invariant polynomials by exactly $|W|$ elements, which allows him to show that such a joint eigenfunction must be a linear combination of terms of the form $e^{\langle w(h_1), h_2 \rangle}$ for $w \in W$.  A clever technique for computing the coefficients in this linear combination then allows him to establish the theorem.

The proof presented here also starts with the theory of radial parts of differential operators, but instead we use this theory to obtain solutions to the heat equation on $\mathfrak{h}_0$ from Ad-invariant solutions to the heat equation on $\mathfrak{g}_0$.  The integral (\ref{eqn:hc}) then appears naturally when the heat kernel on $\mathfrak{g}_0$ is averaged over the adjoint orbits.  This allows us to write the right-hand side of (\ref{eqn:hc}) in terms of a particular solution to the heat equation on $\mathfrak{h}_0$.

The following section consists of the proof of Theorem \ref{thm:hc}.

\section{Proof of the integral formula} \label{sec:hc_heat_eqn_proof}

\subsection{Heat equations on $\mathfrak{g}_0$ and $\mathfrak{h}_0$}

Let $\omega$ be the quadratic Casimir polynomial on $\mathfrak{g}$ defined by $\omega(x) = \langle x, x \rangle$.  The Laplacian on $\mathfrak{g}_0$ is the differential operator $\omega(\partial)$.  A function $\phi: \mathfrak{g}_0 \times (0, \infty) \to \mathbb{C}$ satisfies the heat equation on $\mathfrak{g}_0$ if
\begin{equation} \label{eqn:heat} \left(\partial_t + \frac{1}{2} \omega(\partial_x) \right) \phi(x,t) = 0, \quad x \in \mathfrak{g}_0, \ t \in (0, \infty). \end{equation}
Since $G$ is compact and semisimple, the Killing form is negative definite, so that (\ref{eqn:heat}) contains a plus sign instead of the more familiar minus sign for the Euclidean heat equation.  Similarly, $\psi: \mathfrak{h}_0 \times (0, \infty) \to \mathbb{C}$ satisfies the heat equation on $\mathfrak{h}_0$ if 
\begin{equation} \label{eqn:heat-csa} \left(\partial_t + \frac{1}{2} \bar \omega(\partial_h) \right) \psi(h,t) = 0, \quad h \in \mathfrak{h}_0, \ t \in (0, \infty), \end{equation}
where the bar indicates restriction to $\mathfrak{h}_0$.

We want to establish a relationship between the solutions of (\ref{eqn:heat}) and of (\ref{eqn:heat-csa}).  In general, if $\phi$ satisfies (\ref{eqn:heat}), it is not the case that $\bar \phi$ satisfies (\ref{eqn:heat-csa}).  However, we will show
\begin{lemma} \label{lem:soln-reln} If $\phi \in C^2_1(\mathfrak{g}_0 \times (0, \infty))$ solves (\ref{eqn:heat}) and is invariant under the adjoint action of $G$ in the sense that $$\phi(\mathrm{Ad}_g x, t) = \phi(x, t) \quad \forall g \in G,$$ then $\Pi(h) \bar \phi(h, t)$ solves (\ref{eqn:heat-csa}). \end{lemma}

Before proving the lemma, we need to briefly introduce some notions from Lie theory and geometric analysis.  Define the set of {\it regular elements} of $\mathfrak{h}$ to be $\mathfrak{h}' = \{ h \in \mathfrak{h} \ | \ \Pi(h) \not = 0 \}$.  The submanifold $\mathfrak{h}_0 \cap \mathfrak{h}' \subset \mathfrak{g}_0$ is {\it transverse} to the adjoint orbits in $\mathfrak{g}_0$ in the sense that for each $h \in \mathfrak{h}_0 \cap \mathfrak{h}'$ we have a decomposition of tangent spaces \begin{equation} \label{eqn:transverse}T_h \mathfrak{g}_0  = T_h \mathcal{O}_h \oplus T_h (\mathfrak{h}_0 \cap \mathfrak{h}'),\end{equation} where $\mathcal{O}_h = \{ \mathrm{Ad}_g h \ | \ g \in G \}$ is the adjoint orbit of $h$.  To see this, consider the root space decomposition of $\mathfrak{g}$, $$\mathfrak{g} = \mathfrak{h} \oplus \bigoplus_\alpha \mathfrak{g}_\alpha,$$ where $\alpha$ runs over the roots of $\mathfrak{g}$ with respect to $\mathfrak{h}$.  Under the usual identification $T_h \mathfrak{g}_0 \cong \mathfrak{g}_0$, at $h \in \mathfrak{h}_0 \cap \mathfrak{h}'$ we have $$T_h (\mathfrak{h}_0 \cap \mathfrak{h}') \cong \mathfrak{h}_0, \quad T_h \mathcal{O}_h \cong [ \mathfrak{g}_0, h] = \mathfrak{g}_0 \cap \bigoplus_\alpha \mathfrak{g}_\alpha,$$ which gives the transversality property.

We now state without proof two results from \cite{SH}.  The first is a special case of \cite[ch.~2, Theorem 3.6]{SH}.

\begin{theorem} \label{thm:radial-part} Let $M \subset \mathfrak{g}_0$ be a submanifold of $\mathfrak{g}_0$ that is transverse to the adjoint orbits in the sense of (\ref{eqn:transverse}). Let $D$ be a differential operator\footnotemark \footnotetext{The theorem holds for the general definition of a (linear) differential operator as any $\mathbb{C}$-linear operator $D$ acting on smooth functions and satisfying $\mathrm{supp}(Df) \subseteq \mathrm{supp}(f)$.  See \cite[ch.~2, Theorem 1.4]{SH} for an explanation of why this definition makes sense.} on $\mathfrak{g}_0$.  Then there exists a unique differential operator $\gamma(D)$ on $M$ such that, for each $\mathrm{Ad}$-invariant function $f: \mathfrak{g}_0 \to \mathbb{C}$, $$\overline{(Df)} = \gamma(D) \bar f,$$ the bar indicating restriction to $M$. \end{theorem}

The differential operator $\gamma(D)$ is called the {\it radial part of $D$ with transversal manifold $M$}.  In particular, since $\mathfrak{h}_0 \cap \mathfrak{h}'$ satisfies the transversality condition, we can define the radial part of the Laplacian, $\gamma(\omega(\partial))$, such that $$\overline{\omega(\partial) f} = \gamma(\omega(\partial)) \bar f$$ for all Ad-invariant $f: \mathfrak{g}_0 \to \mathbb{C}$, where the bar indicates restriction to $\mathfrak{h}_0 \cap \mathfrak{h}'$.  The next proposition gives an explicit expression for $\gamma(\omega(\partial))$.  It was originally proven in \cite[Theorem 1]{HC} and also appears as \cite[ch.~2, Proposition 3.14]{SH}.

\begin{proposition} \label{prop:radial-laplacian} The radial part of the Laplacian $\omega(\partial)$ with transversal manifold $\mathfrak{h}_0 \cap \mathfrak{h}'$ is given by $$\gamma(\omega(\partial)) = \Pi^{-1} \bar \omega(\partial) \circ \Pi.$$ \end{proposition}

We now can prove Lemma \ref{lem:soln-reln}.
\begin{proof}[Proof of Lemma \ref{lem:soln-reln}] Restricting (\ref{eqn:heat}) to $\mathfrak{h}_0 \cap \mathfrak{h}'$ and applying Proposition \ref{prop:radial-laplacian} gives the {\it radial heat equation} \begin{equation} \label{eqn:radial-heat} \left( \partial_t + \frac{1}{2} \Pi^{-1}(h) \bar \omega(\partial_{h}) \circ \Pi(h) \right) \bar \phi(h, t) = 0, \quad h \in \mathfrak{h}_0 \cap \mathfrak{h}'.\end{equation} Multiplying both sides by $\Pi$ we have \begin{equation} \label{eqn:rest-heat-soln} \left( \partial_t + \frac{1}{2} \bar \omega(\partial_{h}) \right) \Pi(h) \bar \phi(h, t) = 0, \quad h \in \mathfrak{h}_0 \cap \mathfrak{h}'. \end{equation} Thus $\Pi(h) \bar \phi(h, t)$ solves the heat equation on $\mathfrak{h}_0 \cap \mathfrak{h}'$ and therefore on all of $\mathfrak{h}_0$ by continuity since $\mathfrak{h}_0 \cap \mathfrak{h}'$ is dense.
\end{proof}

\subsection{The $G$-averaged heat kernel}

The {\it heat kernel} on $\mathfrak{g}_0$ is the function $K: \mathfrak{g}_0^2 \times (0, \infty) \to \mathbb{C}$ given by 
\begin{equation} \label{eqn:heat-kernel}
K(x_1, x_2; t) = \left(\frac{1}{2 \pi t} \right )^{\dim \mathfrak{g}_0 /2} e^{\frac{1}{2t} \langle x_1 - x_2, x_1 - x_2 \rangle}.
\end{equation}
It solves the heat equation (\ref{eqn:heat}) where the spatial derivatives act on the $x_2$ variables, with the boundary condition \begin{equation} \label{eqn:bc1} \lim_{t \to 0} K(x_1, x_2; t) = \delta(x_1 - x_2).\end{equation} The limit is understood in the distributional sense of \begin{equation*} \label{eqn:dist-sense} \lim_{t \to 0} \int_{\mathfrak{g}_0} K(x_1, x_2; t) \varphi(x_2) dx_2 = \varphi(x_1) \end{equation*} for $\varphi \in C_c^\infty(\mathfrak{g}_0)$, where $dx_2$ indicates the integration measure induced by the Killing form.  The choice of this integration measure is important, as it guarantees that for all $x_1$ and $t$ we have
$$\int_{\mathfrak{g}_0} K(x_1, x_2; t) dx_2 = 1.$$

Following \cite{IZ}, we define the {\it $G$-averaged heat kernel} as \begin{multline} \label{eqn:k-avg-def} \tilde K(x_1, x_2; t) := \int_G K(\mathrm{Ad}_g x_1, x_2 ; t) dg \\ = \left(\frac{1}{2 \pi t} \right )^{\dim \mathfrak{g}_0 /2} e^{\frac{1}{2t}(\langle x_1, x_1 \rangle + \langle x_2, x_2 \rangle)} I(x_1, -x_2; t),\end{multline} where  $$I(x_1, x_2; t) := \int_G e^{\frac{1}{t} \langle \mathrm{Ad}_g x_1, x_2 \rangle} dg.$$  Observe that for $h_1, h_2 \in \mathfrak{h}_0$, $I(h_1, h_2; 1)$ recovers the integral on the left-hand side of (\ref{eqn:hc}), so that the integral appears naturally in this context.

The function $\tilde K$ is constant on adjoint orbits of both $x_1$ and $x_2$, and by linearity it satisfies the heat equation (\ref{eqn:heat}) on $\mathfrak{g}_0$ as well, so that by Lemma \ref{lem:soln-reln}, $\Pi(h_2) \tilde K(h_1, h_2; t)$ satisfies the heat equation (\ref{eqn:heat-csa}) on $\mathfrak{h}_0$ with the spatial derivatives acting in the $h_2$ variables.  Therefore the function \begin{equation} \label{eqn:v-def} V(h_1, h_2; t) := (2\pi)^{(\dim \mathfrak{g}_0 - N)/2} \Pi(h_1) \Pi(h_2) \tilde K(h_1, h_2; t) \end{equation} also satisfies (\ref{eqn:heat-csa}) and is skew with respect to the action of $W$ on either of $h_1$ or $h_2$ individually. Physically, $V$ can be interpreted as a generalized Slater determinant, with $W$-skewness playing the role of the antisymmetry property of fermions.  In the next step we identify the boundary conditions that $V$ satisfies as $t$ approaches 0.  This will allow us to write an exact expression for $V$ using the fundamental solution to the heat equation on $\mathfrak{h}_0$, yielding (\ref{eqn:hc}).

\subsection{Boundary conditions for $V$}

We now further assume that both $h_1, h_2 \in \mathfrak{h}_0 \cap \mathfrak{h}'$.  In order to compute the distributional limit of $V$ as $t$ approaches 0, we use the steepest-descent method to determine the asymptotics of $I(h_1, h_2; t)$ to leading order in $t$.  We will show:

\begin{lemma} \label{lem:bcs} If $h_1, h_2 \in \mathfrak{h}_0 \cap \mathfrak{h}'$, there exists a constant $C \in \mathbb{R}$ such that \begin{equation} \label{eqn:bcs} \lim_{t \to 0} V(h_1, h_2; t) = C \sum_{w \in W} \epsilon(w) \delta(w(h_1) - h_2) \end{equation} where the distributional sense of the limit is given by integration against test functions in $C_c^\infty(\mathfrak{h}_0)$ with respect to the volume form induced by the restriction of the Killing form to $\mathfrak{h}_0$. \end{lemma}

\begin{proof}
We first rewrite \begin{equation} \label{eqn:rewrite-v} V(h_1, h_2;t) = \frac{t^{-\dim \mathfrak{g}_0 /2}}{(2\pi)^{N/2}} \Pi(h_1) \Pi(h_2) e^{\frac{1}{2t}(\langle h_1, h_1 \rangle + \langle h_2, h_2 \rangle)} I(h_1, -h_2; t).\end{equation}  Next we rewrite $I(h_1, -h_2; t)$ as follows.  Let $S$ be the maximal torus in $G$ with Lie algebra $\mathfrak{h}_0$.  Let $ds$ be the normalized Haar measure on $S$ and let $d(gS)$ be a left-invariant probability measure on $G/S$.  Then by a standard Fubini-type theorem for Lie groups (see e.g. \cite[ch.~1, Theorem 1.9]{SH}) we have \begin{multline} \label{eqn:rewrite-i} I(h_1, -h_2; t) = \int_G e^{\frac{-1}{t} \langle \mathrm{Ad}_g h_1, h_2 \rangle} dg = \int_{G/S} \int_S e^{\frac{-1}{t} \langle \mathrm{Ad}_{gs} h_1, h_2 \rangle} ds \ d(gS) \\ = \int_{G/S} e^{\frac{-1}{t} \langle \mathrm{Ad}_{g} h_1, h_2 \rangle} d(gS).\end{multline}

We will apply the steepest-descent method to the last integral in (\ref{eqn:rewrite-i}).  To do so, we need the following lemmas computing the critical points of the function $gS \mapsto \langle \mathrm{Ad}_g h_1, h_2 \rangle$ along with its Hessian matrix at each critical point.

\begin{lemma} \label{lem:crit-pts}
The critical points of the function $gS \mapsto \langle \mathrm{Ad}_g h_1, h_2 \rangle$ on $G/S$ are the points $gS$ such that $\mathrm{Ad}_g$ represents some $w \in W$ as a linear operator on $\mathfrak{h}$.
\end{lemma}
\begin{proof}
We first note that the map $gS \mapsto \mathrm{Ad}_g h_1$ is a diffeomorphism of $G/S$ onto the adjoint orbit $\mathcal{O}_{h_1} \subset \mathfrak{g}_0$.  Thus we may equivalently find the critical points of the function $x \mapsto \langle x, h_2 \rangle$ for $x \in \mathcal{O}_{h_1}$.  The tangent space at a point $x_0 \in \mathcal{O}_{h_1}$ is $T_{x_0} \mathcal{O}_{h_1}= [x_0, \mathfrak{g}_0],$ and the partial derivative of the linear functional $\langle x, h_2 \rangle$ in the direction of a tangent vector $y$ is equal to $\langle y, h_2 \rangle.$  Thus the condition for $x_0 \in \mathcal{O}_{h_1}$ to be a critical point is $$\langle y, h_2 \rangle = 0 \quad \forall y \in [x_0, \mathfrak{g}_0],$$ or equivalently $$\langle [x_0, y], h_2 \rangle = 0 \quad \forall y \in \mathfrak{g}_0.$$  Using the antisymmetry of the bracket and the invariance of the Killing form, this is equivalent to $$\langle y, [x_0, h_2] \rangle = 0 \quad \forall y \in \mathfrak{g}_0.$$ By the non-degeneracy of the Killing form, this will hold if and only if $[x_0, h_2] = 0$, which implies $x_0 \in \mathfrak{h}_0$.  Writing $x_0 = \mathrm{Ad}_g h_1$, we will have $\mathrm{Ad}_g h_1 \in \mathfrak{h}_0$ exactly when $\mathrm{Ad}_g h_1 = w(h_1)$ for some $w \in W$.  Thus $gS$ corresponds to a critical point exactly when $\mathrm{Ad}_g$ represents an element of the Weyl group.
\end{proof}

\begin{lemma} \label{lem:hessian-det}
Let $H_{xx}$ be the Hessian matrix of the function $gS \mapsto \langle \mathrm{Ad}_g h_1, h_2 \rangle$ at a critical point $g_0 S$, in coordinates on $G/S$ given by $e^xS \mapsto x + \mathfrak{h}_0$.  If $\mathrm{Ad}_{g_0} = w \in W$, then $$\sqrt{\det(-H_{xx})} = \epsilon(w) \Pi(h_1) \Pi(h_2)$$ where the branch of the square root is chosen by writing $\sqrt{\det(- H_{xx})} = \prod_{j=1}^n \sqrt{-\mu_j}$ where $\mu_j$ are the eigenvalues of $H_{xx}$, and taking $|\mathrm{arg} \sqrt{-\mu_j}| < \pi/4$.
\end{lemma}
\begin{proof}
Since $$\mathrm{Ad}_{\exp(x)} = \sum_{j = 0}^\infty \frac{1}{j!} \mathrm{ad}_x^j,$$ expanding to second order in $x$ around the critical point gives \begin{equation} \label{eqn:second-order} \langle \mathrm{Ad}_{\exp(x)} w(h_1), h_2 \rangle = \langle w(h_1), h_2 \rangle + \frac{1}{2}\langle \mathrm{ad}_x^2 w(h_1), h_2 \rangle + O(|x|^3),\end{equation} where $|x| = \sqrt{- \omega(x)}$ is the norm on $\mathfrak{g}_0$.

To obtain the Hessian we must compute explicitly the second-order term in (\ref{eqn:second-order}).  Let $P = \{\alpha_1, \hdots \alpha_r\}$ be the collection of positive roots used to define $\Pi$.  For a root $\alpha$, let $\mathfrak{g}_\alpha \subset \mathfrak{g}$ denote the corresponding root space.  For each $\alpha \in P$, choose $x_\alpha \in \mathfrak{g}_\alpha$, $x_{-\alpha} \in \mathfrak{g}_{-\alpha}$ normalized so that $\langle x_\alpha, x_{-\alpha} \rangle = 1$.  Assuming without loss of generality that $x \in \mathfrak{h}_0^\perp$, we may write \begin{equation} \label{eqn:x_root_decomp} x = \sum_{\alpha \in P} c_\alpha x_\alpha + c_{-\alpha} x_{-\alpha}, \end{equation} and we find by a straightforward calculation that $$\frac{1}{2}\langle \mathrm{ad}_x^2 w(h_1), h_2 \rangle = \sum_{\alpha \in P} \alpha(w(h_1))\alpha(h_2) c_\alpha c_{-\alpha}.$$  Thus $H_{xx}$ is a block-diagonal matrix composed of 2-by-2 blocks of the form $$\begin{bmatrix} 0 & \alpha_i(w(h_1))\alpha_i(h_2) \\ \alpha_i(w(h_1)) \alpha_i(h_2) & 0 \end{bmatrix}.$$  With the appropriate choice of branch for the square root, we find $$\sqrt{\det(-H_{xx})} = \Pi(w(h_1)) \Pi(h_2) = \epsilon(w) \Pi(h_1) \Pi(h_2)$$ as desired.
\end{proof}

Returning now to the proof of Lemma \ref{lem:bcs} and applying the steepest-descent approximation to (\ref{eqn:rewrite-i}) together with Lemmas \ref{lem:crit-pts} and \ref{lem:hessian-det}, we obtain \begin{equation} \label{eqn:i-approx} I(h_1, -h_2; t) = C \frac{\left( 2\pi t \right)^{(\dim \mathfrak{g}_0 -N)/2}}{\Pi(h_1) \Pi(h_2)} \sum_{w \in W} \epsilon(w) e^{\frac{-1}{t}\langle w(h_1), h_2 \rangle}(1 + O(t)) \end{equation} where the constant $C$ arises from the normalization of the measure $d(gS)$.  Substituting this result into (\ref{eqn:rewrite-v}), we find \begin{equation} \label{eqn:v-approx} V(h_1, h_2; t) = C \left( \frac{1}{2 \pi t} \right)^{N/2} \sum_{w \in W} \epsilon(w) e^{\frac{1}{2t} \langle w(h_1) - h_2, w(h_1) - h_2 \rangle}(1 + O(t))\end{equation} as $t \to 0$, which gives the desired limit (\ref{eqn:bcs}).
\end{proof}

Because $V$ solves the heat equation on $\mathfrak{h}_0$, taking the convolution of the boundary data (\ref{eqn:bcs}) with the fundamental solution gives \begin{equation} \label{eqn:exact-v} V(h_1, h_2; t) = C \left( \frac{1}{2 \pi t} \right)^{N/2} \sum_{w \in W} \epsilon(w) e^{\frac{1}{2t} \langle w(h_1) - h_2, w(h_1) - h_2 \rangle}. \end{equation} Thus the higher-order terms on the right-hand side of (\ref{eqn:v-approx}) actually vanish.

\subsection{Normalization} \label{subsec:normalization}

It only remains to rearrange terms and determine the constant $C$.  Evaluating $V$ at $t = 1$, we have
\begin{align*} V(h_1, h_2; 1) &= (2\pi)^{(\dim \mathfrak{g}_0 - N)/2} \Pi(h_1) \Pi(h_2) \tilde K(h_1, h_2; 1) \\
& = (2\pi)^{(\dim \mathfrak{g}_0 - N)/2} \Pi(h_1) \Pi(h_2) \left(\frac{1}{2 \pi} \right )^{\dim \mathfrak{g}_0 /2} e^{\frac{1}{2}(\langle h_1, h_1 \rangle + \langle h_2, h_2 \rangle)} I(h_1, -h_2; 1) \\
& = C \left( \frac{1}{2 \pi} \right)^{N/2} \sum_{w \in W} \epsilon(w) e^{\frac{1}{2} \langle w(h_1) - h_2, w(h_1) - h_2 \rangle}\\
& = C \left( \frac{1}{2 \pi} \right)^{N/2} e^{\frac{1}{2}(\langle h_1, h_1 \rangle + \langle h_2, h_2 \rangle)} \sum_{w \in W} \epsilon(w) e^{\langle w(h_1),-h_2 \rangle}.
\end{align*}
After cancelations, this becomes \begin{equation} \label{eqn:almost-hc} \Pi(h_1) \Pi(h_2) I(h_1, h_2; 1) = C \sum_{w \in W} \epsilon(w) e^{\langle w(h_1),h_2 \rangle}.\end{equation}
Up to this point we have assumed that $h_1, h_2 \in \mathfrak{h}_0 \cap \mathfrak{h}'$, but we can immediately remove this assumption: since both sides of (\ref{eqn:almost-hc}) are analytic in $h_1$ and $h_2$, this identity holds for all $h_1, h_2 \in \mathfrak{h}$.

Finally, we determine $C$.  Applying $\Pi(\partial_{h_1})$ to both sides of (\ref{eqn:almost-hc}) and evaluating at $h_1 = 0$, we obtain $$\Pi(h_2) [ \! [ \Pi, \Pi ] \! ] = C |W| \Pi(h_2), \quad h_2 \in \mathfrak{h},$$
so that $C = [ \! [ \Pi, \Pi ] \! ] / |W|$.  This completes the proof of Theorem \ref{thm:hc}.

\section{Concluding remarks} \label{sec:concl_remarks}

\subsection{Relationship to Calogero--Moser systems} \label{subsec:cm_relationship}
To illustrate the relationship between Harish-Chandra integrals and Calogero--Moser systems, we observe that by Lemma \ref{lem:soln-reln} the function $$ W(h_1, h_2; t) = \frac{1}{N^2} \log \tilde K(\sqrt{N} h_1, \sqrt{N} h_2 ; t) $$ satisfies \begin{equation} \label{eqn:free_energy_pde} 2 \frac{\partial W}{\partial t} = N |\nabla W|^2 + \frac{2}{N} \nabla(\log \Pi ) \cdot \nabla W - \frac{1}{N} \bar \omega(\partial) W - \frac{1}{N^3} \Pi^{-1} \bar \omega(\partial) \Pi\end{equation} where spatial derivatives act in the $h_2$ variables and $\Pi = \Pi(h_2)$.  Since $\Pi$ is harmonic (see \cite[Proposition 4]{CM}), the last term on the right-hand side of (\ref{eqn:free_energy_pde}) vanishes.  To compute the large-$N$ asymptotics for the integral over $U(N)$ in \cite{AM}, Matytsin drops the term $N^{-1} \bar \omega(\partial) W$, arguing heuristically that it should be subdominant as $N \to \infty$.  If we neglect this term and make the substitution $W = S - N^{-2} \log \Pi$, then we arrive at \begin{equation} \label{eqn:CM_preHJ} 2 \frac{\partial S}{\partial t} = N |\nabla S|^2 - N^{-3} |\nabla (\log \Pi) |^2. \end{equation}  In fact (\ref{eqn:CM_preHJ}) is (a scaling in $N$ of) the Hamilton-Jacobi equation for the rational Calogero--Moser system associated to the root system of the algebra $\mathfrak{g}$.  From this observation we expect that the large-$N$ asymptotics of Harish-Chandra integrals can be understood in terms of hydrodynamic scaling limits of Calogero--Moser systems.  This is known to be true for integrals over $U(N)$, as in \cite{AM, GZ} the leading-order asymptotics of the HCIZ integral are derived in terms of a particular solution to the complex Burgers' equation.  As explained in \cite{GM}, the complex Burgers' equation also arises as a hydrodynamic limit of the Calogero--Moser system associated to the $A_N$ root system.

\subsection{Proof of the integral formula via symplectic geometry} \label{subsec:symplectic_proof}
Below we give an alternate proof of Theorem \ref{thm:hc} using a localization technique in symplectic geometry.  While the existence of such a proof was observed already in \cite{DH}, here we present the details of the argument with the goal of making this derivation accessible to non-specialists.

The symplectic geometry approach illustrates a completely different perspective on the Harish-Chandra integral.  Rather than writing the integral in terms of the heat kernel on $\mathfrak{g}_0$, we instead view it as an oscillatory integral over a coadjoint orbit in $\mathfrak{g}_0^*$.  The proof begins with the steepest-descent approximation (\ref{eqn:i-approx}) for $I(h_1, h_2; t)$ obtained in the proof of Lemma \ref{lem:bcs}, which becomes a stationary-phase approximation when $t$ is taken to be imaginary.  Then, instead of observing that $V$ solves the heat equation on $\mathfrak{h}_0$, we use the Duistermaat--Heckman theorem to deduce that the stationary-phase approximation is exact, after which it remains only to compute the normalization constant following the argument of Section \ref{subsec:normalization} above.

We first state the Duistermaat--Heckman theorem, which was proved in \cite{DH} and later shown to be an instance of a more general principle of {\it equivariant localization} \cite{AB}.  Let $(M, \omega)$ be a compact symplectic manifold of dimension $2n$, and let $S$ be a $d$-dimensional torus acting smoothly on $M$, with Lie algebra $\mathfrak{s}$.  For each $s \in \mathfrak{s}$ we define a vector field $X_s$ on $M$ by \begin{equation} \label{eqn:torus_vf_def} X_s(x) f = \frac{d}{dt} \bigg |_{t = 0} f(\exp(ts) \cdot x), \quad x \in M, \ f \in C^\infty(M).\end{equation} We assume that there is a {\it moment map} for the action of $S$ on $M$, that is a smooth function $\Phi: M \to \mathfrak{s}^*$ such that \begin{equation} \label{eqn:moment_map_def} \iota_{X_s} \omega(\cdot) = - ( d\Phi(\cdot), s ), \quad s \in \mathfrak{s},\end{equation} where $( \cdot, \cdot )$ is the duality pairing of $\mathfrak{s}^*$ and $\mathfrak{s}$.\footnotemark \footnotetext{Some authors define the moment map with the opposite sign in (\ref{eqn:moment_map_def}).}  The {\it Liouville measure} $\mu$ on $M$ is given by the volume form $\omega^{\wedge n}/n!$.  Then we have the following:

\begin{theorem}[Duistermaat--Heckman] \label{thm:dh}
The integral \begin{equation} \label{eqn:dh-int} \int_{M} e^{it(\Phi(x), s)} d\mu(x) \end{equation} is exactly equal to its leading-order approximation by the method of stationary phase as $t \to \infty$.\footnotemark \footnotetext{The Duistermaat--Heckman theorem is sometimes stated differently, as follows: every regular value of $\Phi$ has a neighborhood in which $\Phi_*\mu$ is equal to Lebesgue measure times a polynomial of degree at most $n-d$.  Theorem \ref{thm:dh} follows from this statement by observing that the integral (\ref{eqn:dh-int}), considered as a function of $t$, is the inverse Fourier transform of the measure $(\Phi, s)_* \mu$, i.e. the pushforward of $\mu$ by the map $(\Phi(\cdot), s): \mathfrak{s}^* \to \R$.}
\end{theorem}

To prove the exactness of (\ref{eqn:i-approx}) from Theorem \ref{thm:dh}, we re-interpret $I$ in the language of symplectic geometry.  For $x \in \mathfrak{g}_0$ let $x^* \in \mathfrak{g}_0^*$ be its dual under the Killing form, that is $x^*(y) = \langle x, y \rangle$ for $x, y \in \mathfrak{g}_0$.  Let $\beta = h^* \in \mathfrak{h}_0^*$.  We define the {\it coadjoint orbit} $$\mathcal{O}^*_{\beta} = \{ \mathrm{Ad}_g^* \beta \ | \ g \in G \} \subset \mathfrak{g}_0^*,$$ where $\mathrm{Ad}^*$ is the coadjoint representation of $G$ on $\mathfrak{g}_0^*$, defined by $$\mathrm{Ad}^*_g x^* = (\mathrm{Ad}_{g^{-1}} x)^*.$$

At a point $x^* \in \mathcal{O}^*_{\beta}$, under the usual identification $T_{x^*}\mathfrak{g}_0^* \cong \mathfrak{g}_0^*$, we have $$T_{x^*} \mathcal{O}^*_{\beta} = \{ [x, y]^* \ | \ y \in \mathfrak{g}_0^* \}.$$  The {\it Kirillov-Kostant-Souriau form} is the $G$-invariant 2-form $\omega$ on $\mathcal{O}^*_{\beta}$ defined by \begin{equation} \label{eqn:kks_form} \omega_{x^*}([x, y]^*, [x, z]^*) = \langle x, [y, z] \rangle. \end{equation}  This form can be shown by direct computation to be non-degenerate and closed, and it therefore makes $\mathcal{O}^*_{\beta}$ into a symplectic manifold \cite[ch. 1]{AK}.  Let $S \subset G$ be the maximal torus with tangent space $\mathfrak{h}_0$.  Then $S$ acts smoothly on $\mathcal{O}^*_\beta$ via the coadjoint representation.

Let $\Phi: \mathcal{O}^*_\beta \to \mathfrak{h}_0^*$ be the orthogonal projection onto $\mathfrak{h}_0^*$.  We will show that $\Phi$ is a moment map for the action of $S$ on $\mathcal{O}^*_\beta$.  Observe that from the definition of $\Phi$ as an orthogonal projection, we have \begin{equation} \label{eqn:phi_killing} (\Phi(x^*),s) = \langle x, s \rangle, \quad s \in\mathfrak{h}_0,\ x^* \in \mathcal{O}^*_\beta.  \end{equation}  Moreover, $d\Phi([x,y]^*)$ is also just the orthogonal projection of $[x,y]^*$ onto $\mathfrak{h}_0^*$, so that $$-(d\Phi([x,y]^*),s) = - \langle [x,y], s \rangle = \langle x, [s, y] \rangle = \omega_{x^*}([x,s]^*, [x,y]^*).$$ A direct computation from (\ref{eqn:torus_vf_def}) gives $X_s(x^*) = [x,s]^*$, so that $\Phi$ satisfies (\ref{eqn:moment_map_def}) and therefore is a moment map as desired.

Next we relate the Liouville measure on $\mathcal{O}^*_\beta$ to Haar measure on $G$.  Let $2n = \dim \mathcal{O}^*_{\beta}.$  The Liouville measure $\omega^{\wedge n}/n!$ is $G$-invariant due to the invariance of $\omega$.  Pulling the Liouville measure back along the map $\mathrm{Ad}^* \beta: G \to \mathcal{O}^*_{\beta}$, we obtain a finite invariant measure on $G$, which by the uniqueness of Haar measure must equal a constant times $dg$.  Thus for a Borel set $E \subset G$ we have

\begin{equation} \label{eqn:liouville_haar_measures}
\int_E dg = \frac{1}{\mathrm{vol}_\mu(\mathcal{O}^*_\beta)} \int_{\{ \mathrm{Ad}^*_g \beta \ | \ g \in E \}} \frac{\omega^{\wedge n}}{n!}.
\end{equation}

Putting together (\ref{eqn:phi_killing}) and (\ref{eqn:liouville_haar_measures}) and writing $d\mu = \omega^{\wedge n}/n!$, we can express $I(h_1, h_2; t)$ as an integral over $\mathcal{O}_{h_1^*}^*$:
$$I(h_1, h_2; t) = \int_G e^{\frac{1}{t} \langle \mathrm{Ad}_g h_1, h_2 \rangle} dg = \frac{1}{\mathrm{vol}_\mu(\mathcal{O}^*_{h_1^*})} \int_{\mathcal{O}_{h_1^*}^*} e^{\frac{1}{t} (\Phi(\beta), h_2)} d\mu(\beta).$$  This function is analytic in $t$ for $t \in (0, \infty)$, and so by analytic continuation we can equivalently consider $$I(h_1, h_2; -it^{-1}) = \frac{1}{\mathrm{vol}_\mu(\mathcal{O}^*_{h_1^*})} \int_{\mathcal{O}_{h_1^*}^*} e^{it (\Phi(\beta), h_2)} d\mu(\beta).$$
We have now written $I$ in the form (\ref{eqn:dh-int}), so that by Theorem \ref{thm:dh} and (\ref{eqn:i-approx}) we have $$I(h_1, h_2; -it^{-1}) = C \left( \frac{2\pi i}{t} \right)^{(\dim \mathfrak{g}_0 -N)/2}(\Pi(h_1) \Pi(h_2))^{-1} \sum_{w \in W} \epsilon(w) e^{it \langle w(h_1), h_2 \rangle}$$ for some constant $C$. The proof concludes by evaluating at $t = -i$ and computing $C$ as in Section \ref{subsec:normalization}.

\subsection{Relationship to irreducible characters} \label{subsec:repthy_proof}

The Harish-Chandra integral also has an interpretation in terms of the irreducible characters of $G$.  Here we discuss how Theorem \ref{thm:hc} is essentially equivalent to the Kirillov character formula in the case of a compact, connected, semisimple group.

Let $\lambda \in i \mathfrak{h}_0^*$ be the highest weight of an irreducible representation of $G$ with character $\chi_\lambda$, and let $\rho = \frac{1}{2}\sum_{\alpha \in P} \alpha$.  Let $\mu$ be the Liouville measure associated to the Kirillov-Kostant-Souriau form on the coadjoint orbit $\mathcal{O}^*_{-i(\lambda+\rho)}$.  The {\it Kirillov character formula} for compact groups \cite[ch. 5, Theorem 9]{AK} says:

\begin{theorem} \label{thm:kirillov_char}
For $h \in \mathfrak{h}_0 \cap \mathfrak{h}'$,
\begin{equation} \label{eqn:kirillov_char} \chi_\lambda(e^h) = \frac{\Pi(h)}{\prod_{\alpha \in P} (e^{\alpha(h)/2} - e^{-\alpha(h)/2})} \int_{\mathcal{O}^*_{-i(\lambda + \rho)}} e^{i(\beta, h)} d\mu(\beta). \end{equation}
\end{theorem}

Although here we assume that $G$ satisfies the assumptions of Theorem \ref{thm:hc}, versions of this formula hold in a variety of situations even for non-compact groups; see the book \cite{AK} for a detailed discussion.

On the other hand, the {\it Weyl character formula} (see e.g. \cite[Theorem 25.4]{DB}) states:
\begin{theorem} \label{thm:weyl_char}
For $h \in \mathfrak{h}_0 \cap \mathfrak{h}'$,
\begin{equation} \label{eqn:weyl_char}
\chi_\lambda(e^h) = \frac{\sum_{w \in W} \epsilon(w) e^{(w(\lambda + \rho), h)}}{\prod_{\alpha \in P} (e^{\alpha(h)/2} - e^{-\alpha(h)/2})}.
\end{equation}
\end{theorem}

Note that either of (\ref{eqn:kirillov_char}) or (\ref{eqn:weyl_char}) completely determines $\chi_\lambda$. Since the character is analytic it is determined on the maximal torus $\exp(\mathfrak{h}_0)$ by its values on $\exp(\mathfrak{h}_0 \cap \mathfrak{h}')$, and since it is a class function it is determined on all of $G$ by its restriction to a maximal torus.

One standard proof of Theorem \ref{thm:kirillov_char} dating to Kirillov's paper \cite{AK2} works by applying Theorem \ref{thm:weyl_char} to the right-hand side of (\ref{eqn:hc}). However, the Kirillov formula can also be proven by other methods, for example by applying the equivariant index theorem to the Dirac operator twisted by the Kostant--Souriau line bundle on the coadjoint orbit \cite{BGV}. Nothing stops us, therefore, from using (\ref{eqn:kirillov_char}) to prove (\ref{eqn:hc}) instead, as there is no circularity involved. For the sake of completeness, we record this proof below. \\

Equating the right-hand sides of (\ref{eqn:kirillov_char}) and (\ref{eqn:weyl_char}), writing $h = h_2$ and $\lambda + \rho = ih_1^*$, and using the relation (\ref{eqn:liouville_haar_measures}) between Liouville measure on $\mathcal{O}^*_{h_1^*}$ and Haar measure on $G$, we obtain \begin{multline*} \Pi(h) \int_{\mathcal{O}^*_{-i(\lambda + \rho)}} e^{i(\beta, h)} d\mu(\beta) \\ = \mathrm{vol}_\mu(\mathcal{O}^*_{h_1^*}) \Pi(h_2) \int_G e^{i \langle \mathrm{Ad}_g h_1, h_2 \rangle} dg = \sum_{w \in W} \epsilon(w) e^{i \langle w(h_1), h_2 \rangle}.\end{multline*}

Applying $\Pi(\partial_{h_2})$ to both sides and evaluating at $h_2 = 0$, we find \begin{equation} \label{eqn:coadjoint_volume} \mathrm{vol}_\mu(\mathcal{O}^*_{h_1^*}) = \frac{|W|} {[\! [ \Pi, \Pi ] \! ]} \Pi(h_1),\end{equation} which matches the normalization in (\ref{eqn:hc}). It remains to remove the assumptions that $h_2 \in \mathfrak{h}_0 \cap \mathfrak{h}'$ and $i h_1^* = \lambda + \rho$ for $\lambda$ a dominant weight.  Since the right-hand side of (\ref{eqn:hc}) is $W$-invariant in $h_1$, the result also holds for $h_1$ such that $ih_1^* = w(\lambda + \rho)$, with $w \in W$ and $\lambda$ a dominant weight.  Such $h_1$ form a lattice spanning $\mathfrak{h}_0$, so by scaling $h_2$ and shifting the scaling onto $h_1$ we obtain the result for all $h_2 \in \mathfrak{h}_0 \cap \mathfrak{h}'$ and $h_1$ in a dense subset of $\mathfrak{h}_0$.\footnotemark \footnotetext{This type of approximation argument was introduced by Ziegler in his proof of the Kostant convexity theorem in \cite{FZ}.} Analytic continuation then gives the result for all $h_1, h_2 \in \mathfrak{h}$, completing the proof of Theorem \ref{thm:hc}.

\section*{Acknowledgements}
The author acknowledges partial support from NSF grants DMS 1714187 and DMS 1148284. The author also thanks Govind Menon, Jean-Bernard Zuber, Pierre Le Doussal, Peter Forrester, Sigurdur Helgason and Boris Hanin for helpful comments and conversations, as well as the Park City Mathematics Institute, supported by NSF grant DMS 1441467, for the opportunity to participate in the 2017 summer school on random matrices.

\bibliographystyle{alpha}

\end{document}